\documentclass{mcom-l}

\usepackage{amssymb}
\newtheorem{theorem}{Theorem}[section]
\newtheorem{lemma}[theorem]{Lemma}

\theoremstyle{definition}
\newtheorem{definition}[theorem]{Definition}
\newtheorem{proposition}[theorem]{Proposition}

\theoremstyle{remark}
\newtheorem{remark}[theorem]{Remark}

\numberwithin{equation}{section}

\usepackage[colorlinks=true, urlcolor = blue, citecolor=black]{hyperref}
\usepackage{xcolor}
\usepackage{physics}

\usepackage[justification=centerlast]{caption}
\usepackage[labelformat=simple]{subcaption}

\usepackage{graphicx}

\begin{document}

\title{Finding hypergraph immersion is fixed-parameter tractable}

\author{Xiangyi~Meng}%
\address{Department of Physics, Applied Physics, and Astronomy, Rensselaer Polytechnic Institute, Troy, New York 12180, USA}%
\curraddr{}
\email{xmenggroup@gmail.com}
\thanks{}

\author{Yu~Tian}%
\address{Nordita, Stockholm University and KTH Royal Institute of Technology, Stockholm, Sweden}
\curraddr{Center for Systems Biology Dresden, Dresden, Germany}
\email{yu.tian.research@gmail.com}
\thanks{}

\subjclass[2024]{Primary 05C83 Secondary 05C85, 05C65}

\dedicatory{}

\begin{abstract}
Immersion minor is an important variant of graph minor, defined through an injective mapping from vertices in a smaller graph $H$ to vertices in a larger graph $G$ where adjacent elements of the former are connected in the latter by edge-disjoint paths. Here, we consider the immersion problem in the emerging field of hypergraphs. 
We first define hypergraph immersion by extending the injective mapping to hypergraphs. We then prove that finding a hypergraph immersion is fixed-parameter tractable, namely, there exists an $O(N^6)$ polynomial-time algorithm to determine whether a fixed hypergraph $H$ can be immersed in a hypergraph $G$ with $N$ vertices.
Additionally, we present the dual hypergraph immersion problem and provide further characteristics of the algorithmic complexity. 
\end{abstract}

\maketitle

\section{Introduction}

Graph minor is an important notion in graph theory. 
A graph\footnote{\footnotesize We consider finite, undirected graph, where multiedges may be present but loops are ignored, unless otherwise stated.} 
$H$ is referred to as a minor of the graph $G$ if $H$ can be formed by deleting edges, vertices, and by contracting edges, i.e.,~removing an edge while simultaneously merging two vertices it used to connect. 
The theory of graph minors can be dated back to one of the most fundamental theorems in graph theory, the Wagner’s theorem, where a graph is planar if and only if it does not have the complete graph $K_5$ or the complete bipartite graph $K_{3,3}$ as minors~\cite{wagner1937minor}. 
Notably, Robertson and Seymour have developed the theory of graph minors and their variants in a series of over 20 papers spanning 20 years. 
Specifically, as in the Wagner’s theorem, their results imply that analogous forbidden minor characterisation exists for every property of graphs that is preserved by deletions and edge contractions~\cite{robertson-seymour_rs10} (see \cite{lovasz2005minor} for a survey). 
Computationally, for every fixed graph $H$, it is possible to test whether $H$ is a minor of an input graph $G$ in polynomial time~\cite{robertson1995paths}. Together with the forbidden minor characterisation, this result may be extended to every graph property preserved by deletions and contractions~\cite{fellows1988linear}. 

Immersion minor is an important variant of graph minors. 
A graph $H$ is referred to as an immersion minor of the graph $G$ if there exists an injective mapping from vertices in $H$ to vertices in $G$, where images of $H$'s neighboring vertices are connected in $G$ by edge-disjoint paths. 
The immersion problem plays a prominent role in various settings. 
In graph drawing, immersion minor arises as the planarisation of non-planar graphs, which allows drawing methods for planar graphs to be extended to non-planar graphs~\cite{buchheim2014drawing}. In various combinatorial problems, such as determining the cutwidth and its multidimensional generalisations, and congestion problems, immersion is the key to solving them in linear time.
One can also think of circuit fabrication, parallel computation, and network design, and we refer the reader to~\cite{fellows1992app} for more discussion.   

In particular, the relationship between immersion minor and quantum operations has attracted recent research interest~\cite{tian2024quanimmersion}. In quantum networks, vertices that are not directly connected by an edge can still be entangled through a fundamental operation known as entanglement routing 
\cite{q-netw-route_lphnm20,q-netw-route_pkttjbeg19,q-netw-route_p19,multipartite-q-netw-route_sb23}, provided that there is a path of edges connecting the vertices. The operation of entanglement routing, however, utilises all the edges along the path that connect the vertices, and subsequently deletes the utilised edges from the graph, leading to a dynamic process called path percolation~\cite{path-percolation_mhrk24}. This edge-disjoint nature of entanglement routing (namely, each edge can only be utilised once) enables an exact mapping to immersion minor, 
allowing us to reformulate entanglement routing in quantum information as a graph-theoretical problem.

Unlike in ordinary graphs, where edges connect exactly two vertices, hyperedges can connect more than two vertices, providing a natural graph-equivalence of multipartite entangled states.  These states, with their inherent higher-order correlations, are crucial building blocks for future quantum network design.
\cite{multipartite-q-netw_ctpv21,multipartite-q-netw_hpe19,multipartite-commun-q-netw,multipartite-q-netw_wxkhlhgsp20}. Hence, hypergraphs, comprising vertices and hyperedges, provide a natural framework in the quantum setting. 
Not only being mathematically interesting and necessary in quantum information processing, hypergraphs are also more adequate for the modeling of a variety of empirical interaction data, such as in chemistry~\cite{jost2019hyperL}, biology~\cite{klamt2009biology}, and social sciences 
\cite{bianconi2021hyper,krumov2011collab,Kwang2023core,taramasco2010collab}. The research field of hypergraphs has boomed in the past decades, both theoretically and empirically. 
However, to our best knowledge, the notion of immersion minor has not been explored in hypergraphs. 

{In this paper, we fill in this gap by first extending the definition of immersion minor, together with the relevant operations on graphs, to the general setting of hypergraphs. This requires us to decompose the single operation (``lifting'') that defines immersion minor in ordinary graphs to {two operations} (which we term ``coalescence'' and ``dewetting'') in hypergraphs. The two operations equivalently define hypergraph immersion and reduce to the lifting operation for ordinary graphs.
The main result of our paper is to show that a polynomial algorithm exists whose time complexity is in the order of a fixed power of the size of the larger hypergraph $G$. In other words, the hypergraph immersion problem is {fixed-parameter tractable}. 
This provides significant theoretical support to the graph immersion algorithms and also other relevant algorithms in the case of hypergraphs, such as the algorithms for graph embedding and downstream applications in areas including quantum communication. 
Furthermore, we introduce the concept of dual hypergraph immersion, defined as immersion between the transpose of hypergraphs, showcasing the generality of our results.}

Our paper is organised as follows. In section~\ref{sec:preliminary}, we introduce basic concepts and the definition of immersion minor in ordinary graphs, and then discuss important notions in hypergraphs. In section~\ref{sec:main_results}, we present two equivalent definitions of immersion in hypergraphs, one through edge-disjoint paths extended for hypergraphs, and the other through immersion operations on hypergraphs. We then give the proof structure of our main result that finding hypergraph immersion is fixed-parameter tractable. In section~\ref{sec:proof_ordinary}, we first prove the case when the immersed graph $H$ is an ordinary graph. In section~\ref{sec:proof_hyper}, we extend the proof to the immersed graph $H$ being a general hypergraph, where the extra complexity is handled by considering the finiteness of topological classes for general hypergraphs. Finally, in section~\ref{sec:dual}, we introduce the dual hypergraph representation and provide further characteristics for the problem of dual immersion for hypergraphs.

\section{Preliminary}\label{sec:preliminary}
In this section, we introduce the basic concepts and the graph immersion problem, followed by key hypergraph notions.

\subsection{Immersion problem in (ordinary) graphs}
A graph consists of vertices and edges connecting them. We denote by $V(G)$ and $E(G)$ for the set of vertices and edges, respectively. 
In this paper, we primarily consider finite, undirected graphs, which may contain multiedges between vertices unless explicitly stated as \emph{simple graphs} (i.e.,~without multiedges). Loops (edges starting and ending at the same vertex) are excluded from our analysis.
We may explicitly use the term \emph{ordinary graph} to distinguish from hypergraphs, emphasising the fact that each edge connects precisely two vertices in an ordinary graph; that is, $\forall e\in E(G)$, $e= \{u, v\}$ with $u,v\in V(G)$.

\begin{definition}[immersion in ordinary graphs~\cite{robertson-seymour_rs10}]
    Let $G, H$ be loopless graphs. An \textit{immersion} of $H$ in $G$ is a function $\alpha$ with domain $V(H)\cup E(H)$ such that: 
\begin{enumerate}
    \item $\alpha(v) \in V(G)$ for all $v\in V(H)$, and $\alpha(u)\ne\alpha(v)$ for all distinct $u,v\in V(H)$;
    \item for each edge $e\in E(H)$, if $e$ has distinct ends $u,v$, then $\alpha(e)$ is a path of $G$ with ends $\alpha(u), \alpha(v)$; 
    \item for all distinct $e_1,e_2\in E(H)$, $E(\alpha(e_1)\cap \alpha(e_2))=\emptyset$, where $E(\alpha(e_1)\cap \alpha(e_2))$ denotes the common edges of $\alpha(e_1)$ and $\alpha(e_2))$.
\end{enumerate}
\end{definition}
Immersion can be equivalently defined through the \textit{lifting operation} on graphs:
\begin{definition}
    An immersion of $H$ in $G$ exists if and only if $H$ can be obtained from (a subgraph of) $G$ by \textit{lifting}, an operation on adjacent edges:
    \begin{itemize}
        \item Given three vertices $v$, $u$, and $w$, where $\{v,u\}$ and $\{u,w\}$ are edges in the graph, the lifting of $vuw$, or equivalently of $\{v,u\}, \{u,w\}$ is the operation that deletes the two edges $\{v,u\}$ and $\{u,w\}$ and adds the edge $\{v,w\}$. 
    \end{itemize}
    \label{def:immersion}
\end{definition}
\begin{remark}
    In the case where $\{v,w\}$ was already present, $v$ and $w$ will now be connected by more than one edge. Hence this operation is intrinsically a multi-graph operation.
\end{remark}

By changing the edge-disjoint paths to vertex-disjoint paths, we obtain another important notion of embedding in graphs. 
\begin{definition}[embedding in ordinary graphs~\cite{graph-theor-appl}]
    Let $G, H$ be loopless graphs. An \emph{embedding}, or topological minor, of $H$ in $G$ is a function $\beta$ with domain $V(H)\cup E(H)$ such that: 
    \begin{enumerate}
        \item $\beta(v) \in V(G)$ for all $v\in V(H)$, and $\beta(u)\ne\beta(v)$ for all distinct $u,v\in V(H)$;
        \item for each edge $e\in E(H)$, if $e$ has distinct ends $u,v$, then $\beta(e)$ is a path of $G$ with ends $\beta(u), \beta(v)$; 
        \item for all distinct $e_1,e_2\in E(H)$, $V(\beta(e_1)\cap \beta(e_2))=\beta(V(e_1 \cap e_2))$, where $V(e_1 \cap e_2)$ denotes the common endpoints of $e_1$ and $e_2$.
    \end{enumerate}
    \label{def:embedding}
\end{definition}
Similarly, graph embedding can also be equivalently defined through the \textit{subdivision operation} on graphs. 
\begin{definition}
    An embedding of $H$ in $G$ exists if and only if (a subgraph) of $G$ can be obtained from $H$ by \textit{subdivision}, an operation on edges:
    \begin{itemize}
        \item Given an edge $\{u,v\}$, where $u,v$ are vertices in the graph, the subdivision of $\{u,v\}$ is the operation that adds a new vertex $w$ and replaces the original edge $\{u,v\}$ by two new edges $\{u,w\}$ and $\{w,v\}$.
    \end{itemize}
    \label{def:embedding-subdivision}
\end{definition}

\begin{remark}
    From Definitions~\ref{def:immersion} and~\ref{def:embedding}, we can immediately see that an embedding is an immersion, but not vice versa.
\end{remark}
As an important breakthrough in the algorithmic study of graph minors, it has been shown that the running time to check the existence of an embedding of a graph $H$ in a graph $G$ is fixed-parameter tractable with $H$ as the parameter. 
This means the algorithm's time complexity, for varying $G$, is $O(\left|V(G)\right|^{O(1)})$, where the $O(1)$ term does not increase with the size of $H$ asymptotically. Specifically, we have:
\begin{theorem}[time complexity of embedding in ordinary graphs~\cite{embed-fix-param-tract_gkmw11}]
    For every graph $H$, there is an  $O(\left|V(G)\right|^3)$ time algorithm that decides if $H$ has an embedding in a graph $G$.
    \label{the:embedding-time}
\end{theorem}

\subsection{Hypergraph}
A \emph{hypergraph} relaxes the constraints on the number of vertices in each edge, where a hyperedge $e$ can connect an arbitrary number of 
vertices, and we denote the size of a hyperedge by the number of vertices in this edge. An $r$-uniform hypergraph is then a hypergraph where each hyperedge has size $r$. A complete $r$-uniform hypergraph is then the one with all possible edges of size $r$, or the edge set consisting of all size-$r$ subsets of the vertex set, denoted by $K_{n}^r$ where $n$ denotes the size of the hypergraph, i.e.,~$n=\left|V(K_{n}^r)\right|$. 

The end-to-end connection in hypergraphs can be considered through the generalisation of a path in an ordinary graph to hypergraphs, known as the \emph{Berge path}~\cite{berge1973hyper}. Specifically, a Berge path of length $t$ is an alternating sequence of distinct $t+1$ vertices and distinct $t$ hyperedges of the hypergraph $G$, $v_1, e_1, v_2, e_2, v_3,\cdots, e_t, v_{t+1}$, such that $v_i, v_{i+1}\in e_i$, for $i=1,\cdots,t$. 
A Berge path is a natural generalisation of a path in an ordinary graph. A \emph{Berge cycle} is the same as a Berge path except that the last vertex $v_{t+1}$ is replaced by $v_1$, making the sequence cyclic.
Further, a \emph{connected subgraph} of a hypergraph is a set of vertices that are pairwise connected by some Berge paths. The size of the subgraph is the size of its vertex set.

\subsection{Factor graph}
A factor graph $F(G)$ of a hypergraph $G$ is a ``larger'' (but ordinary) 
graph that still maintains the topology of $G$ to some extent. The factor graph is defined as follows~\cite{hypergr-percolation_bd23}:
Each vertex $\nu\in V(G)$ corresponds to a vertex $\nu'\in V(F(G))$, and
each hyperedge $\mu\in E(G)$ also corresponds to a vertex $\mu'\in V(F(G))$. In other words, $\left|V(F(G))\right|=\left|V(G)\right|+\left|E(G)\right|$. 
The topological connectivity in $G$ is maintained through edges in $F(G)$, where edges exist (and only exist) between the vertices corresponding to vertices in $G$ and those corresponding to hyperedges in $G$.
Specifically, 
a hyperedge $\mu$ is incident to a vertex $\nu$ in $G$ if and only if there is an edge between the corresponding vertices $\mu',\nu'\in V(F(G))$.

\section{Main results}\label{sec:main_results}
In this section, we first extend the immersion problem to hypergraphs, including the corresponding operations in hypergraphs. We then show that finding hypergraph immersion is fixed-parameter tractable and give a sketch of our proofs that are detailed in the following sections. 

\subsection{Hypergraph immersion}\label{sec:hypergraph}
We first generalise the notion of immersion to hypergraphs and propose the following \emph{hypergraph immersion}, where the key requirement of edge-disjoint paths is maintained but through connected subgraphs (or essentially Berge paths) in hypergraphs:

\begin{definition}
\label{definition_1}
Let $G, H$ be loopless hypergraphs\footnote{\footnotesize We define a loop in a hypergraph as a hyperedge containing a vertex more than once. By definition, size-$1$ hyperedges containing only a single vertex are permitted.}. An \emph{immersion} of $H$ in $G$ is a function $\alpha$ with domain $V(H)\cup E(H)$, such that:
\begin{enumerate}
    \item $\alpha(v)\in V(G)$ for all $v\in V(H)$, and $\alpha(v_1)\neq \alpha(v_2)$ for all distinct $v_1,v_2\in V(H)$; 
    \item for each hyperedge $e\in E(H)$, if $e$ has distinct ends $v_1,v_2,\cdots$, then $\alpha(e)$ is a connected subgraph in $G$ that includes $\alpha(v_1),\alpha(v_2),\cdots$. 
    \item for all distinct $e_1,e_2\in E(H)$, $E(\alpha(e_1)\cap \alpha(e_2))=\emptyset$;
\end{enumerate}
\end{definition} 

In other words, there is an injective mapping from vertices in $H$ to vertices in $G$ and from hyperedges in $H$ to \textit{edge-disjoint} connected subgraphs in $G$. 
As discussed in section~\ref{sec:preliminary}, immersion in ordinary graphs can be equivalently defined through the lifting operation. 
We now give the alternative definition of immersion through hypergraph operations in Definition~\ref{definition_2}. As we will see in Proposition~\ref{pro:definition_equiv}, it is necessary to decompose the lifting operation for ordinary graphs into two operations for hypergraphs, the \textit{coalescence} and \textit{dewetting} operations:
\begin{definition}
\label{definition_2}
    We define two operations (Fig.~\ref{fig_operation}):
    \begin{itemize}
        \item \emph{(Edge) Coalesce}: Merge two hyperedges that share at least one vertex, resulting in a new hyperedge that is incident to all vertices originally incident to the two hyperedges. Specifically, given two hyperedges $e_1, e_2$ with $e_1\cap e_2 \ne \emptyset$, the (edge) coalescence of $e_1, e_2$ is the operation that deletes the two edges $e_1, e_2$ and add the edge $e_1\cup e_2$. 
        \item \emph{Dewet}: Detach a hyperedge from one vertex that the hyperedge is incident on. Specifically, given a hyperedge $e$ and a vertex $v\in e$, the dewetting of $e$ from $v$ is the operation that deletes $e$ and adds the edge $e\backslash\{v\}$.
\end{itemize}
\end{definition}
\begin{remark}
    It is straightforward to see that the lifting operation in ordinary graphs can be obtained from the two operations in Definition~\ref{definition_2}, by first applying the edge coalescence to two adjacent edges and then dewetting the newly formed edge from the shared vertex. 
    The nature of hypergraphs adds to the complexity of the immersion problem, which makes it necessary to decompose the lifting operation into two operations.  
\end{remark}

\begin{proposition}[Theorem~1 in Ref.~\cite{tian2024quanimmersion}]
    An immersion $\alpha$ of a hypergraph $H$ in a hypergraph $G$, as defined in Definition~\ref{definition_1}, exists if and only if $H$ can be obtained from a subgraph of $G$ by a sequence of coalescence and dewetting operations. 
    \label{pro:definition_equiv}
\end{proposition}
\begin{proof}
(Also see Ref.~\cite{tian2024quanimmersion})
If an immersion $\alpha$ exists, we can show that $H$ is isomorphic to a hypergraph $G^{\text{cl;dw}}$ obtained by applying the coalescence and dewetting operations on a subgraph of $G$, where the bijection corresponding to the isomorphism $f: V(H)\to V(H')$ takes the same value as $\alpha$, i.e.,~$\forall v\in V(H)$, $f(v) = \alpha(v)$. 
    
    We build $G^{\text{cl;dw}}$ as follows. 
    For each $e\in E(H)$ with distinct ends $v_1, v_2, \dots, v_{\left|e\right|}$, we apply the coalescence operation to every adjacent pair of hyperedges in the corresponding connected subgraph $\alpha(e)\subseteq G$, until we obtain one single hyperedge incident to all vertices in $\{\alpha(v_1), \alpha(v_2), \dots, \alpha(v_{\left|e\right|})\}$. The coalescence operation is possible since $\alpha(e)$ is connected.
    Then, we apply the dewetting operation to detach the single hyperedge from vertex $u$, $\forall u\in V(\alpha(e))\backslash \{\alpha(v_1), \alpha(v_2), \dots, \alpha(v_{\left|e\right|})\}$. Doing this for each edge-disjoint subgraph $\alpha(e)$, $\forall e\in E(H)$, we denote the resulted hypergraph, a collection of single hyperedges derived from every $\alpha(e)$, by $G^{\text{cl;dw}}$.
    It is straightforward to show that $H$ is isomorphic to $G^{\text{cl;dw}}$ with the bijection $f$. 
    
    Now, if $H$ can be obtained by applying the two operations to a subgraph of $G$, we can immediately construct $\alpha$ which maps $v\in V(H)$ to $\alpha(v)\in V(G)$, the one it comes from. Then $a(v_1)\ne \alpha(v_2)$ for all distinct $v_1,v_2\in V(H)$. Also, for each edge $e\in E(H)$, we can retrieve the set of edges $\alpha(e)\in E(G)$ where we apply the operations to obtain $e$. Then $\alpha(e)$ is connected since the two operations can only be applied to connected edges. Also, for all distinct $e_1, e_2\in E(H)$, $E(\alpha(e_1)\cap\alpha(e_2)) = \emptyset$, since each edge in $G$ can be used at most once in the operations. Hence, $H$ is immersed in $G$ by definition. 
\end{proof}

\begin{figure}[h!]
    \centering
    \begin{minipage}[b]{120pt}
        \begin{minipage}[b]{120pt}
		\centering
		{\includegraphics[width=120pt]{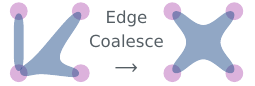}\subcaption{\label{fig_coalesce}}}
	\end{minipage} \\[20pt] 
 
    \begin{minipage}[b]{120pt}
		\centering
		{\includegraphics[width=120pt]{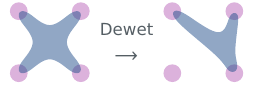}\subcaption{\label{fig_evaporate}}}
	\end{minipage}
    \end{minipage}\hspace{2mm}
    \begin{minipage}[b]{98pt}
            \centering
            {\includegraphics[width=98pt]{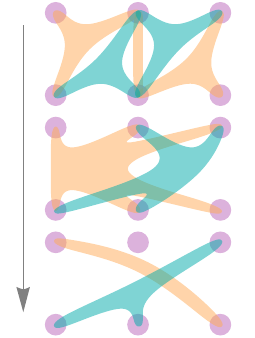}\subcaption{\label{fig_coa_eva}}}
    \end{minipage}
    \vspace{-3mm}
    \caption{\textbf{Hypergraph immersion.} (a) The ``edge coalescence'' operation involves merging two connected hyperedges, resulting in a new hyperedge that encompasses all vertices from the original pair. (b) The ``dewetting'' entails the removal of a single vertex from a hyperedge, effectively breaking its connection with that vertex. 
    (c) Transformation of QN from an initial topology to a final one, achieved through a series of operations. The final topology is effectively immersed within the original.\vspace{-2mm}
    \hfill\hfill}
    \label{fig_operation}
\end{figure}

\subsection{Fixed-parameter tractable}
We prove the existence of a polynomial algorithm $\sim O(\left|V(G)\right|^6)$, where $V(G)$ is the number of vertices in $G$, to decide whether $H$ can be immersed in $G$. Notably, the exponent $6$ we discover is a constant, independent of $H$. This makes the algorithm \emph{fixed-parameter tractable} for fully fixed $H$. 
Note that the same question of whether immersion in an ordinary graph is fixed-parameter tractable was first raised in 1992 by Downey and Fellows~\cite{fix-param-tract_df92}, but only answered with positivity after two decades~\cite{embed-fix-param-tract_gkmw11}. Our result, as a generalisation of Ref.~\cite{embed-fix-param-tract_gkmw11}, now extends both the question and answer to the hypergraph realm.

\begin{theorem}
\label{theorem_main}
	For every (ordinary or hyper) graph $H$, there is an $O(\left|V(G)\right|^6)$ time algorithm that decides if $H$ can be immersed in a hypergraph $G$.
 \label{the:fixed-params}
\end{theorem}

Before giving the detailed proof in the following sections, we sketch the key idea of the proof. Essentially, we hope to transform the hypergraphs $H$ and $G$ into their respective factor graphs, $F(H)$ and $F(G)$, which are ordinary graphs. This seemingly reduces the hypergraph immersion problem to a standard graph immersion problem. However, two key challenges arise:

\begin{enumerate}
    \item In the transformation to factor graphs, both hyperedges and vertices of $H$ become vertices in $F(H)$, and similarly for $G$ and $F(G)$. Consequently, a vertex-to-vertex mapping between $F(H)$ and $F(G)$ does not guarantee a corresponding mapping between $H$ and $G$, potentially violating criterion~(1) of Definition~\ref{definition_1}.
    \item Criterion~(2) of Definition~\ref{definition_1} requires mapping each hyperedge $e \in E(H)$ to a connected subgraph $\alpha(e) \subseteq G$ without specifying $\alpha(e)$'s topology. However, it is worth noting that converting to $F(H)$ imposes a star-graph topology for each $e \in E(H)$. Thus, even if $e$ can be immersed in $G$ with some (potentially non-star) topology, the corresponding star graph in $F(H)$ might not be immersible in $F(G)$ if $F(G)$ lacks the required star topology.
\end{enumerate}

We will address these challenges as follows:
\begin{enumerate}
    \item We adopt the technique in Ref.~\cite{embed-fix-param-tract_gkmw11}, replacing each vertex in both $H$ and $G$ with a sufficiently large complete graph. This ``densification'' ensures that a complete graph in $H$ can only be immersed in a corresponding complete graph in $G$. With sufficiently large complete graphs, this establishes a rigorous vertex-to-vertex mapping between $H$ and $G$. This technique was originally used to convert ordinary graph immersion to an embedding problem while preserving vertex mapping~\cite{embed-fix-param-tract_gkmw11}. As we demonstrate, the technique is also applicable to hypergraphs, enabling the conversion of hypergraph immersion to an ordinary graph embedding problem (using factor graphs), and leveraging existing results on graph embedding for our proof.
    \item Rather than considering the factor graph of a single hypergraph $H$, we examine the factor graphs of a set of hypergraphs $\Tilde{H}$ as variants of $H$, each we call a \emph{division} of $H$. A division $\Tilde{H}$ is constructed by replacing each hyperedge $e\in E(H)$ with a connected subgraph whose factor graph has a possibly different, non-star topology (which is essentially a spanning tree) than the standard star topology of the factor graph of $e$. Since the number of spanning trees connecting a finite set of terminals is finite, we can build a finite number of divisions $\Tilde{H}$ encompassing all possible topologies for the immersion of each $e \in E(H)$. Testing the factor graphs of these divisions provides a criterion for determining whether $G$ contains a topology capable of immersing one of the divisions, thus completing the proof.
\end{enumerate}

\section{Proof for ordinary graph $H$}\label{sec:proof_ordinary}
For simplicity, we first consider the case when $H$ is a simple graph:
\begin{theorem}
\label{theorem_1}
	For every ordinary, simple graph $H$, there is an $O(\left|V(G)\right|^6)$ time algorithm that decides if $H$ can be immersed in a hypergraph $G$.
\end{theorem}
The proof structure of Theorem~\ref{theorem_1} generalises the discussion on immersion for ordinary, simple graphs by Grohe~et al.~\cite{embed-fix-param-tract_gkmw11}.
Specifically, we will represent each hyperedge in hypergraph $G$ as a vertex in a new, ordinary graph, such that the original immersion problem in a hypergraph may be reformulated as an embedding problem in the corresponding ordinary graph, as in Definition~\ref{def:embedding}. We term this new graph as an \textit{$M$-generalised factor graph}.

\subsection{$M$-generalised factor graph}\label{sec:M-factor-G}
An $M$-generalised factor graph $F_M(G)$ 
generalises the factor graph $F(G)$. Here,
we denote the $M$-generalised factor graph of $G$ by $G_{M}'$.
each hyperedge $\mu\in E(G)$ corresponds to a vertex $\mu'\in V(G_{M}')$ and each vertex $\nu\in V(G)$ corresponds to $M$ duplicate vertices $\nu_1',\nu_2',\cdots,\nu_M'\in V(G_{M}')$, with the subscript $i$ in $\nu_i'$ to separate different copies for $i\in\{1,\dots, M\}$, $M\in\mathbb{N}_+$ and assumed to be a large number; see Fig.~\ref{fig_subdivision_G} for an example when $M=3$.  
Every $G_{M}'$ is an ordinary and simple graph of size $\left|V(G_{M}')\right|=M\left|V(G)\right|+\left|E(G)\right|$. 
Similar to the factor graph formalism, edges exist (and only exist) between the vertices corresponding to vertices in $G$ and those corresponding to hyperedges in $G$.
Specifically, 
a hyperedge $\mu$ is incident to a vertex $\nu$ in $G$ if and only if the corresponding vertex $\mu'\in V(G_{M}')$ is adjacent to all vertices $\nu_1',\nu_2',\cdots,\nu_M'\in V(G_{M}')$. 
In particular, when $M=1$, we retrieve the factor graph of $G$, i.e.,~$F_1(G)\equiv F(G)$.
\begin{figure}[t!]
	\centering
    \includegraphics[width=243pt]{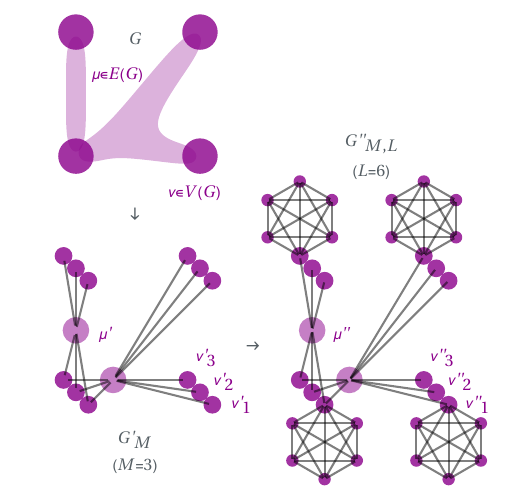}
    \caption{\textbf{Converting hypergraph $G$ to ordinary graphs.}  The ordinary graph $G_{M}'$ is the $M$-generalised factor graph of $G$, and $G_{M,L}''$ is the densified version of $G_{M}'$.
    \hfill\hfill}
    \label{fig_subdivision_G}
\end{figure}
\begin{lemma}
\label{lemma_1_1}
    If an ordinary, simple graph $H$ has an immersion $\alpha$ in a hypergraph $G$, then $H$ has an embedding $\beta$ in the $M$-generalised factor graph $G_{M}'$ given that $M$ is sufficiently large. 
\end{lemma}
\begin{proof}
To show this, let $\beta$ match the same mapping of $\alpha$, i.e.,~if we denote $\alpha(v)$ for some $v\in V(H)$ as a vertex $\nu\in V(G)$, then 
w.l.o.g.~let $\beta(v)=\nu_1'\in V(G_{M}')$. 

To match the connected subgraphs, note that for all $e\in E(H)$, their corresponding connected subgraphs $\alpha(e)$ in $G$ are edge-disjoint. Therefore, each hyperedge $\mu\in E(G)$ is used at most once when forming the connected subgraphs. But this also means each corresponding vertex $\mu'$ in $G_{M}'$ will be used at most once when forming the connected subgraphs $\beta(e)$ in $G_{M}'$. The rest of vertices, $\nu_2',\cdots,\nu_M'\in V(G_{M}')$ will also be used at most once, provided that $M$ is large enough to accommodate as many as $\left|E(H)\right|$ different connected subgraphs that may pass through $\nu$. Hence, the connected subgraphs $\beta(e)$ in $G_{M}'$ are vertex disjoint; thus, $\beta$ is an embedding.
\end{proof}

However, the converse of Lemma~\ref{lemma_1_1} is not always true: a graph $H$ may not have an immersion in $G$ even if $H$ has an embedding $\beta$ in $G_{M}'$. This can occur when $\beta$ maps a vertex in $H$  to (i) some vertex $\mu'\in E(G) \subset V(G_{M}')$ that corresponds to a hyperedge $\mu \in E(G)$, or (ii) some vertex $\nu_i'\in 
V(G_{M}')$, $i\neq 1$, which conflicts with another differently mapped vertex $\nu_1' \in V(G_{M}')$, with both $\nu_1'$ and $\nu_i'$ ($i\neq 1$) corresponding to the same vertex $\nu \in V(G)$. 
To solve this mismatching, it is necessary to restrict every vertex $v\in V(H)$ to be mapped only to one of the vertex duplicates, w.l.o.g.~the first one: $\beta(v)=\nu_1'\in V(G_{M}')$, which always originates from a different vertex $\nu \in V(G)$ for different $v$. 
This restriction is crucial for establishing a sufficient and necessary condition to convert the immersion problem to an embedding problem, and we achieve it by \textit{densifying} the graph, as we introduce below.

\subsection{Densifying a graph}
\label{sec:proof1-densify}
We construct a new graph as follows. 
We replace each 
$\nu_1' \in V(G_{M}')$
by a large, ordinary, complete graph $K_L$ (i.e.,~a complete $2$-uniform $K_{n=L}^2$), and identify one of the $L$ vertices in $K_L$ as $\nu_1'$, in the sense that the connection between $\nu_1'$ and other vertices in $G_{M}'$ are transferred to the connections between this vertex and the corresponding vertices in the newly constructed graph; see Fig.~\ref{fig_subdivision_G} for an example when $L=6$.
This operation necessarily ``densifies'' $G_{M}'$, yielding a new graph $G_{M,L}''$ with 
$|V(G_{M,L}'')|=|V(G_{M}')|+(L-1)|V(G)|$; see Fig.~\ref{fig_subdivision_G} 
for an example. 
Similarly, we also densify $H$, by replacing each $v\in V(H)$ by a complete graph $K_L$ and identifying one of the $L$ vertices in $K_L$ as $v$. 
This yields a new graph $H_{L}''$ with 
$|V(H_{L}'')|=L|V(H)|$. 
Note that now every vertex in $H_{L}''$ has a degree no smaller than $L-1$, and so do the $|V(G)|$ vertices in $G_{M,L}''$ that have been ``densified'' in the construction. 

\begin{lemma}
\label{lemma_1_2}
Given sufficiently large $M$ and $L$, an ordinary, simple graph $H$ has an immersion $\alpha$ in a hypergraph $G$ if and only if
the densified graph $H_{L}''$ has an embedding $\beta$ in the densified, $M$-generalised factor graph $G_{M,L}''$.
\end{lemma}
\begin{proof}
The proof of the necessary condition (``only if'') is the same as that for embedding $H$ in $G_{M}'$, as in the proof of Lemma~\ref{lemma_1_1}. Indeed, note that every new $L$-clique $K_L$ introduced in $H_{L}''$ for each vertex $v \in V(H)$ can be embedded into the corresponding new $L$-clique introduced in $G_{M,L}''$. This shows that $H_{L}''$ can be embedded in $G_{M,L}''$ the same way as $H$ can be embedded in $G_{M}'$.

The proof of the sufficient condition (``if'') goes as follows. Given that every vertex $v\in V(H)$ is identified with a vertex $v''$ in some $L$-clique in $H_{L}''$, the degree of $v''$ must be strictly larger than $L-1$ in $H_{L}''$ (given $L-1$ edges from the $L$-clique and extra incident edges originally from $H$). Hence, $\beta(v'')\in V(G_{M,L}'')$ also has degree more than $L-1$ in $G_{M,L}''$, and thus for sufficiently large $L$, the embedding $\beta$ cannot map $v''$ to $\mu''$ nor $\nu_2'',\cdots,\nu_M''\in V(G_{M,L}'')$,  of which the degrees do not depend on $L$. Similarly, $\beta$ cannot map $v''$ to any vertex of degree exactly $L-1$ that 
comes with the $L$-cliques introduced to $G_{M,L}''$. Hence, $\beta$ can map $v''$ only to $\nu_1''\in V(G_{M,L}'')$ that is identified with $\nu_1'\in V(G_{M}')$. The restriction is thus successfully imposed.
As the connected subgraphs $\beta(e'')$ in $G_{M,L}''$ for every $e''\in E(H_{L}'')$ that comes from $e\in E(H)$ are vertex disjoint, these connected subgraphs cannot go inside the $L$-cliques in $G_{M,L}''$. This means that $\beta$ is also an embedding that maps $H$ to $G_{M}'$.
Now that $\beta$ is restricted to mapping $v\in V(H)$ only to $\nu_1'\in V(G_{M}')$, we can introduce a function $\alpha$ that follows the same mapping of $\beta$, mapping each $v\in V(H)$ to a different vertex $\nu\in V(G)$. As the connected subgraphs $\beta(e)$ in $G_{M}'$ for every $e\in E(H)$ are vertex disjoint, the connected subgraphs $\alpha(e)$ in $G$ are guaranteed edge-disjoint. Hence, $\alpha$ is an immersion.
\end{proof}

\begin{remark}
To ensure that $M$ and $L$ are sufficiently large, they must be bounded by 
$M>|E(H)|$ and $L> M|V(G_{M}')|$, respectively. 
This yields the upper bounds 
$M=O(|V(G)|^0)$ and $L=O(|V(G_{M}')|^1)$. 
We note that in the case of $G$ being an ordinary graph, the upper bound for $L$ can be reduced to $L=O(\left|V(G_{M}')\right|^0)$~\cite{embed-fix-param-tract_gkmw11}, which is however not the case when $G$ can have hyperedges.
\end{remark}

From Theorem~\ref{the:embedding-time}, we know that the running time to check if $H_{L}''$ has an embedding in $G_{M,L}''$ is $O(\left|V(G_{M,L}'')\right|^3)$ for ordinary, simple graphs $H_{L}''$ and $G_{M,L}''$. 
The last step of proving Theorem~\ref{theorem_1} thus necessitates expressing 
$|V(G_{M,L}'')|$ in terms of $|V(G)|$. 
Recall that 
$|V(G_{M,L}'')| = |V(G_{M}')|+(L-1)|V(G)|$ where $\left|V(G_{M}')\right|=M\left|V(G)\right|+\left|E(G)\right|$. 
The following generalisation of Mader's theorem~\cite{mader_m67} puts an additional constraint on $\left|E(G)\right|$:
\begin{lemma}[Mader's theorem for hypergraphs]
\label{lemma_1_3}
Let $\left|E(G)\right|=C\left|V(G)\right|$ for a hypergraph $G$. Let $H$ be an ordinary, simple graph. If $C=C(H)$ is a sufficiently large constant that depends only on $H$, then $G$ has enough connectivity that allows $H$ to be immersed.
\end{lemma}
\begin{proof}
The original Mader's theorem~\cite{mader_m67} is applicable when $H$ and $G$ are ordinary, simple graphs. Its extension to hypergraphs $G$ with multi-hyperedges can be done by the following: First, reduce $G$, using the dewetting operation in Definition~\ref{definition_2}, to an ordinary graph with the same number of vertices and (ordinary) edges; then, remove all but $\left|E(H)\right|$ multiedges per pair of vertices from the ordinary graph. The resulting graph, denoted $G_0$, has at most $\left|E(H)\right|$ multiedges per pair of vertices.

Let $\left|E(G_0)\right|/\left|E(H)\right|=C_0  \left|V(G_0)\right|$. 
By Mader's theorem, If $C_0=C_0(H)$ is a sufficiently large constant that depends only on $H$, then $G_0$ has enough connectivity that allows $H$ to be immersed. Furthermore, this also indicates that $H$ can be immersed in $G$, since by definition $G_0$ is automatically immersed in $G$.
\end{proof}

Finally, we give the complete proof of Theorem~\ref{theorem_1} as follows. 
\begin{proof}[Proof of Theorem~\ref{theorem_1}]
    From Lemma~\ref{lemma_1_2}, to decide if $H$ can be immersed in $G$ is equivalent to determining if the densified graph $H_L''$ has an embedding in the densified, M-generalised factor graph $G_{M,L}''$. 
    From Theorem~\ref{the:embedding-time}, there is an $O(|V(G_{M,L}''|^3)$ time algorithm for the latter. From the construction of $G_{M,L}''$, we know that $|V(G_{M,L}'')| = |V(G_{M}')|+(L-1)|V(G)|$ where $\left|V(G_{M}')\right|=M\left|V(G)\right|+\left|E(G)\right|$. Hence, 
    \begin{align*}
        \left|V(G_{M,L}'')\right| = M\left|V(G)\right|+\left|E(G)\right| + (L-1)\left|V(G)\right|.
    \end{align*}

    The proof of Lemma~\ref{lemma_1_2} implies $M=O(|V(G)|^0)$ and $L=O(|V(G_{M}')|^1)$. From Lemma~\ref{lemma_1_3}, $\left|E(G)\right|$ is, at most, $O(\left|V(G)\right|)$. Hence, 
    \begin{align*}
        \left|V(G_{M,L}'')\right|=O(\left|V(G)\right|^2)
    \end{align*}
    Taken together, the running time for testing immersion is $O(\left|V(G)\right|^6)$.
\end{proof}

\section{Proof for hypergraph $H$}
\label{sec:proof_hyper}
Now we consider the case when $H$ is a hypergraph:
\begin{theorem}
\label{theorem_2}
	For every hypergraph $H$, there is an $O(\left|V(G)\right|^6)$ time algorithm that decides if $H$ can be immersed in a hypergraph $G$.
\end{theorem}

The proof structure of Theorem~\ref{theorem_2} further generalises the proof of Theorem~\ref{theorem_1} to the case when $H$ can have hyperedges. 
However, since $H$ is a hypergraph, its embedding is not defined, and hence Theorem~\ref{theorem_1} is not directly applicable. A seemingly natural approach to overcome this, mirroring the treatment of $G$, is to construct the factor graph of $H$ (which is an ordinary graph) and consider the embedding of this factor graph (as well as its ``densified'' version).
We denote the factor graph of $H$ by $H'$.
each hyperedge ${e}\in E(H)$ corresponds to a vertex ${e}'\in V(H')$, and each vertex ${v}\in V(H)$ is identified with ${v}'\in V(H')$. Let all vertices in $H'$ be isolated from each other unless, if hyperedge $e$ is incident to $v$ in $H$, then vertex $e'$ is adjacent to $v'$ in $H'$.

Unfortunately, the fact that $H$ has an immersion in $G$ does not necessarily lead to that the factor graph $H'$ can be embedded in $G_{M}'$; see Fig.~\ref{fig_counterexample} for an example. 
As a result, Lemma~\ref{lemma_1_1} does not directly apply to hypergraphs, necessitating a new approach (Lemma~\ref{lemma_division_set}, as we will see) to prove Theorem~\ref{theorem_2}.

The underlying reason why Lemma~\ref{lemma_1_1} is inapplicable is that while immersion only considers whether vertices are connected (Definition~\ref{definition_1}), embedding emphasises on \emph{how} they are connected topologically. This distinction becomes crucial for hyperedges. Consider a hyperedge connecting $r$ vertices. These vertices can be connected in various ways, following different spanning tree structures. These structures are not topologically equivalent to the hyperedge's factor graph representation, which always has a star topology (Fig.~\ref{fig_counterexample}).

\begin{figure}[t!]
	\centering
    \includegraphics[width=243pt]{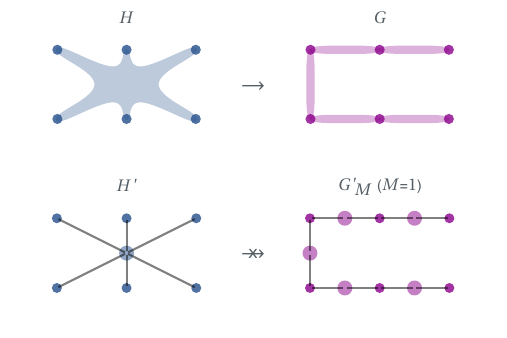}
    \caption{\textbf{Hypergraph $H$ and its factor graph $H'$.} While $H$ can be immersed in $G$, the factor graph $H'$ may not be embedded in $G_{M}'$.\hfill\hfill}
    \label{fig_counterexample}
\end{figure}
To circumvent this, we consider the factor graph(s) of not $H$, but a larger (yet finite) set of hypergraphs, $\mathcal{D}(H)$, as we will introduce below. 

\begin{figure}[t!]
	\centering
    \begin{minipage}[b]{52pt}
		\centering
		{\includegraphics[width=243pt]{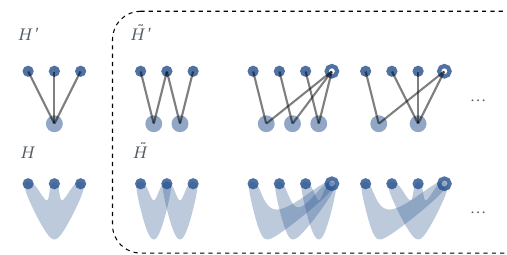}\subcaption{\label{fig_division_hyperedge_a}}}
	\end{minipage}
    \begin{minipage}[b]{52pt}
		\centering
		\subcaption{\label{fig_division_hyperedge_b}}
	\end{minipage}
    \begin{minipage}[b]{52pt}
		\centering
		\subcaption{\label{fig_division_hyperedge_c}}
	\end{minipage}
    \begin{minipage}[b]{52pt}
		\centering
		\subcaption{\label{fig_division_hyperedge_d}}
	\end{minipage}
    \caption{\textbf{Division of a hyperedge.} Given a hyperedge $e$ of size $r$ (here $r=3$), let $H=\{e\}$ and its factor graph be $H'$. Replace $H'$ by a factor graph $\Tilde{H}'$ of some hypergraph $\Tilde{H}$, such that  $\Tilde{H}'$ is a Steiner tree with the $r$ vertices in $\Tilde{H}$ as the terminals. This hypergraph $\Tilde{H}$ is defined as a division of $e$. While the number of possible divisions of $e$ is infinite, there are only finite topological classes of those divisions. For example, \subref{fig_division_hyperedge_c}~and~\subref{fig_division_hyperedge_d} [but not~\subref{fig_division_hyperedge_b}] belong to the same topological class of~\subref{fig_division_hyperedge_a}, as their factor graphs are equivalent up to subdivisions.\hfill\hfill}
    \label{fig_division_hyperedge}
\end{figure}

\subsection{Division of a hyperedge}
Suppose $H$ is a single hyperedge $e$ of size $r$, i.e.,~$E(H)=\{e\}$ and $\left|V(H)\right|=r$. Its factor graph, $H'$, is a bipartite graph between the two vertex partitions, $V(H')=V(H)\cup E(H)$. 
Instead of only the factor graph, we will also consider some other bipartite graphs, each denoted by $\Tilde{H}'$.
We know that there is a bijective correspondence between factor graphs of hypergraphs and bipartite graphs, given the correspondence between the vertex, edge sets in the hypergraphs, and the two parts in the bipartite graphs. Hence, $\Tilde{H}'$ must be the factor graph of some hypergraph. We denote this hypergraph w.l.o.g.~by $\Tilde{H}$.
Specifically, we require that 
$\Tilde{H}'$ must be a \emph{Steiner tree} with the $r$ vertices 
in $H$ as terminals, i.e., $\Tilde{H}'$ is a tree that connects the $r$ vertices such that removing any edge will disconnect the graph (since we consider unweighted graphs in this paper).

\begin{lemma}
    A Steiner tree of an unweighted bipartite graph with $r$ vertices in the same partition as terminals has at most $r$ degree-$1$ vertices. 
    \label{lem:division-e-r-leaves}
\end{lemma}
\begin{proof}
    We denote the two partitions of the bipartite graph as $V_1, V_2$, and w.l.o.g., we assume the $r$ vertices are in $V_1$. 
    We first note that a Steiner tree cannot contain vertices in $V_1$ that are not the $r$ terminals and have degree $1$, since they only connect with vertices in $V_2$, and removing them does not affect the connectivity of the terminals. Hence, there are at most $r$ vertices of degree $1$ in $V_1$.
    
    Now suppose by contradiction that there are $r+1$ vertices of degree $1$, then there is at least one vertex in $V_2$ of degree $1$. However, removing this vertex will not affect the connectivity of the $r$ terminals, hence contradicting the graph being a Steiner tree. 
\end{proof}
A \textit{division of a hyperedge} $e$ is defined by replacing the hyperedge $e$ with the hypergraph corresponding to a Steiner tree, $\Tilde{H}$; see Fig.~\ref{fig_division_hyperedge}. Note that this replacement may introduce additional vertices and hyperedges; thus, $\Tilde{H}$ can be larger than $H$. 

\subsection{Division of a hypergraph}
\label{sec:division-H}
A \textit{division of a hypergraph} $H$ is defined as the hypergraph $\Tilde{H}$ such that all hyperedges in $H$ are replaced by their divisions. By definition, $H$ is a division of itself.

\begin{remark}
By definition, each hyperedge $e\in E(H)$ corresponds to a connected subgraph in $\Tilde{H}$ (from its factor graph being connected). 
Then we can define a mapping from $H$ to $\Tilde{H}$, $\alpha$, such that $\alpha(e)$ returns the corresponding connected subgraph, and for each $v\in V(H)$, $\alpha(v)$ returns the corresponding vertex in $\Tilde{H}$.
It is straightforward to check that $\alpha$ satisfies the three requirements of immersion in Definition~\ref{definition_1}. Thus, $H$ has an immersion in every $\Tilde{H}$.
\end{remark}

\begin{lemma}
\label{lemma_division_set}
    If $H$ has an immersion $\alpha$ in $G$, then there exists at least one division $\Tilde{H}$ whose factor graph $\Tilde{H}'$ has an embedding $\beta$ in the $M$-generalised factor graph $G'_M$ given that $M$ is sufficiently large.
\end{lemma}
\begin{proof}
For each hyperedge $e \in E(H)$, consider its corresponding connected subgraph $\alpha(e)$ in $G$ under the immersion. Due to its connectivity, the factor graph of $\alpha(e)$ necessarily contains a Steiner tree connecting all vertices incident to $e$. This Steiner tree induces a division of $e$, resulting in a division $\tilde{H}$ of $H$. Hence, following the approach in the proof of Lemma~\ref{lemma_1_1}, we can utilise the $M$ copies of each vertex in $G'_M$; and for sufficiently large $M$, we can independently construct vertex-disjoint divisions for each hyperedge in $\tilde{H}'$, yielding an embedding $\beta$ of $\tilde{H}'$ in $G'_M$.
\end{proof}

Hence, Lemma~\ref{lemma_division_set} (in parallel to Lemma~\ref{lemma_1_1}) can be used to establish the necessary condition of Lemma~\ref{lemma_1_2} when extended to hypergraph $H$. 
In the case where $H$ is a hypergraph, we need to consider not only $H$, but also all possible divisions $\Tilde{H}$, and check whether their factor graphs $\Tilde{H}'$ have an embedding in $G'_M$ until we find one that satisfies the condition or all graphs have been exhausted. 
The complexity of the algorithm depends on the size of all possible divisions that are necessary to be tested. 
In general, there can be infinitely many divisions of hypergraph $H$ if there is no constraint. However, as we will show below, the testing set of $\Tilde{H}'$ can be effectively reduced to be of finite size which depends only on $H$, via the notion of \emph{topological equivalence}. 

\subsection{Topological equivalence of divisions}
We define two hypergraphs as topologically equivalent if and only if the factor graph of one hypergraph can be obtained from the factor graph of the other with the subdivision operation as defined in Definition~\ref{def:embedding-subdivision}. It is straightforward to show that the topological equivalence is an equivalence relation, and we refer to a set of hypergraphs that are topologically equivalent as a \emph{topological class}.

\begin{lemma}
    \label{lemma_division_finite}
Given a finite hypergraph $H$, there is only a finite number of topological classes for all divisions of $H$.
\end{lemma}
\begin{proof}
    By the definition of division in section~\ref{sec:division-H}, for every hyperedge $e\in E(H)$, its division's factor graph is a Steiner tree of some bipartite graph.  
    By Lemma~\ref{lem:division-e-r-leaves}, the factor graph has at most $r = |e|$ leaves, and in a similar manner, we can show that it has at most $r-2$ vertices of degrees larger than $2$. 
    This leaves the rest, possibly infinitely many vertices, only have degree $2$, and thus can be removed without disturbing the topological equivalence. Specifically, this is achieved by the inverse of graph subdivision, which is referred to as the \textit{smoothing} operation: for a vertex $w$ of degree $2$ and two neighbors $u,v$, the smoothing is to remove vertex $w$ and replace the two edges $\{u,w\}$ and $\{w,v\}$ by a new edge $\{u,v\}$. 
    The result, after exhausting all possible subdivisions (while preserving the bipartite-graph labeling of the tree), is still a Steiner tree of a bipartite graph with the $r$ vertices in $e$ as the terminals, but with only a finite number of vertices. 
    Therefore, there are only a finite number of such trees of different topologies. Each tree corresponds to a different topological class.

    Since $H$ only has a finite number of hyperedges, the overall number of topological classes of $H$, which is a subset of the direct product of different classes for each hyperedge, is still finite.
\end{proof}
\begin{lemma}
    There exists a finite, minimum set $\mathcal{D}(H)$ of all divisions of $H$ that cover all their possible topological classes.
    \label{lem:hyperH-division-finite}
\end{lemma}
\begin{proof}
    It is straightforward from Lemma~\ref{lemma_division_finite}.
\end{proof}

\begin{lemma}
\label{lemma_2_1}
    If $H$ has an immersion $\alpha$ in $G$, then there exists at least one division $\Tilde{H}\in \mathcal{D}(H)$ whose factor graph $\Tilde{H}'$ has an embedding $\beta$ in the $M$-generalised factor graph $G'_M$ given that $M$ is sufficiently large.
\end{lemma}
\begin{proof}
By Lemma~\ref{lemma_division_set}, we know that there exists one division that satisfies the condition. 
By Definition~\ref{def:embedding-subdivision}, if the factor graph of some division has an embedding $\beta$ in $G'_M$, then its reduced version, by applying the inverse of graph subdivision, also has an embedding in $G'_M$. Therefore, one only needs to consider the embedding of the most reduced version, which must be the factor graph of some division $\Tilde{H}\in \mathcal{D}(H)$ that is included in the minimum set.
\end{proof}

\begin{figure*}[t!]
	\centering
    \hspace*{-5em}
    \includegraphics[width=397.2pt]{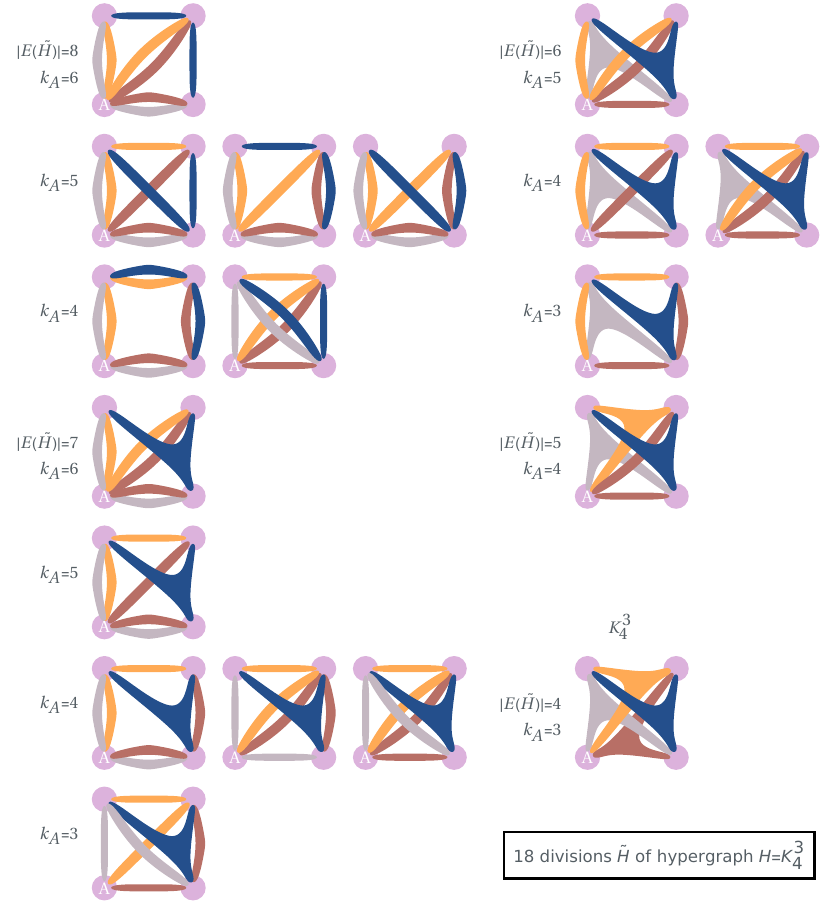}
    \caption{\textbf{The minimum set $\mathcal{D}(H)$.} For $H=K_4^3$, the minimum set $\mathcal{D}(K_4^3)$ includes $18$ divisions $\Tilde{H}$, each belonging to a different topological class. Including all topological classes ensures that, for every $G$ in which $H$ has an immersion, there is at least one $\Tilde{H}\in\mathcal{D}(H)$ of which the factor graph $\Tilde{H}'$ has an embedding in the $M$-generalised factor graph $G_{M}'$ (Fig.~\ref{fig_subdivision_G}). \hfill\hfill}\label{fig_variants}
\end{figure*}
\begin{remark}
    Lemma~\ref{lemma_2_1} indicates that we just need to select a finite subset of all possible divisions of $H$ for testing.  
    For example, if $H$ is an ordinary graph, the minimum set $\mathcal{D}(H)$ contains only one $\Tilde{H}$, which is simply $H$. Another example is the $3$-uniform complete hypergraph of size $4$, $K_4^3$. The minimum set $\mathcal{D}(K_4^3)$ can be effectively identified by replacing any number of hyperedges in $K_4^3$ with two size-$2$ edges connecting the three vertices in any order.
    This results in $18$ topological classes $\Tilde{H}\in \mathcal{D}(K_4^3)$, each with a unique topology up to graph isomorphism; see Fig.~\ref{fig_variants}. 
\end{remark}

\begin{figure}[t!]
	\centering
    \includegraphics[width=243pt]{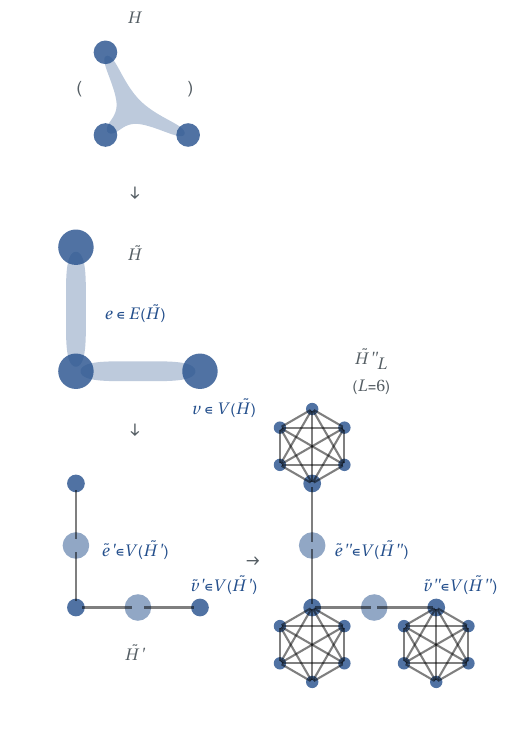}
    \caption{\textbf{Converting hypergraph $H$ to graphs.} The first step is to pick a division of $H$, given by $\Tilde{H}\in\mathcal{D}(H)$. Then, the simple graphs $\Tilde{H}'$ and $\Tilde{H}_{L}''$ are derived from $\Tilde{H}$ (similar to Fig.~\ref{fig_subdivision_G} but not the same).\hfill\hfill}
    \label{fig_subdivision_H}
\end{figure}

\subsection{Densifying graphs}
Now, we examine the complexity of testing the embedding of a factor graph $\Tilde{H}'$ in $G_M'$.
Specifically, we construct new, densified graphs for every $\Tilde{H}'$ as follows.  
We replace each $\Tilde{v}'\in V(\Tilde{H}')$ that corresponds to an original vertex $v\in V(H)$ by a complete graph $K_L$, and identify one of the vertices of $K_L$ as $\Tilde{v}'$, in the sense that all the connections of $\Tilde{v}'$ in $\Tilde{H}'$ are transferred to this vertex in the newly constructed graph. 
This yields a new densified graph $\Tilde{H}_{L}''$ with 
$|V(\Tilde{H}_{L}'')|=|V(\Tilde{H}')|+(L-1)|V(H)|$. 
Note that when defining the division of $H$, new vertices may be added to $\Tilde{H}$ [e.g.,~the open circles in Figs.~\ref{fig_division_hyperedge_c}~and~\ref{fig_division_hyperedge_d}]. Such vertices will \emph{not} be densified, as these vertices do not correspond to the ``real'' original vertices in $H$ and therefore are allowed to overlap with each other in immersion. 
Thus, there are $|V(H)|$ vertices in $\Tilde{H}_{L}''$ that must have a degree no smaller than $L-1$. The rest vertices in $\Tilde{H}'$, corresponding to (1) $E(\Tilde{H}) \subset V(\Tilde{H}')$ and (2) $V(\Tilde{H})\setminus V(H)\subset V(\Tilde{H}')$, are unchanged during the densification.

Correspondingly, we also densify the $M$-generalised factor graph of $G$, $G_{M}'$, to obtain a new graph $G_{M,L}''$ for further exploration, as in section~\ref{sec:proof1-densify}.

\begin{lemma}
\label{lemma_2_2}
Given sufficiently large $M$ and $L$, a hypergraph $H$ has an immersion $\alpha$ in a hypergraph $G$ if and only if there exists $\Tilde{H}\in\mathcal{D}(H)$ such that the densified factor graph
$\Tilde{H}_{L}''$ has an embedding $\beta$ in $G_{M,L}''$.
\end{lemma}
\begin{proof}
To prove the necessary condition (``only if''), we know that from Lemma~\ref{lemma_2_1}, there exists $\tilde{H}\in \mathcal{D}(H)$ whose factor graph $\tilde{H}'$ has an embedding in $G_{M}'$. Then, $\Tilde{H}_{L}''$ can be embedded in $G_{M,L}''$ the same way as $\Tilde{H}'$ in $G_{M}'$.

To prove the sufficient condition (``if''), following the same argument in the proof of Theorem~\ref{theorem_1}, we can show that $\beta$ can only map $\Tilde{v}''\in V(\Tilde{H}_{L}'')$, that is identified with a vertex $v\in V(H)$, 
to $\nu_1''\in V(G_{M,L}'')$ that is identified with $\nu_1'\in V(G_{M}')$. 
This allows us to show further that the edges in $\Tilde{H}_{L}''$ that come from $\Tilde{H}'$ (i.e.,~not the new edges introduced within the $L$-cliques) cannot be embedded inside the $L$-cliques in $G_{M,L}''$.
Specifically, due to the Steiner-tree nature, every edge in the factor graph $\Tilde{H}'$ must be part of a path between two vertices of $\Tilde{H}$ that are identified with different vertices in $V(H)$. Since different vertices $v\in V(H)$ are mapped to different $\nu_1''\in V(G_{M,L}'')$, the aforementioned path (together with its edges) cannot be mapped inside the $L$-cliques in $G_{M,L}''$.
Hence, only $L$-cliques in $\Tilde{H}_{L}''$ are mapped to $L$-cliques in $G_{M,L}''$, meaning that $\beta$ is also an embedding that maps $\Tilde{H}'$ to $G_{M}'$.

The remaining vertices in $\Tilde{H}'$, other than those identified with $v\in V(H)$, are unrestricted in their mapping and thus
may be mapped to the vertex duplicates, $\nu_2',\cdots,\nu_M'\in V(G_{M}')$. 
Hence, we may not find an immersion of $\Tilde{H}$ or $\Tilde{H}'$ in $G$, but we can still find an immersion of $H$ in $G$ directly. 
Indeed, since $\beta$ is restricted to mapping $v\in V(H)$ only to $\nu_1'\in V(G_{M}')$, we have no trouble finding an immersion of $H$ in $G$ the same way as in the proof of Lemma~\ref{lemma_1_2}. 
\end{proof}

Now, we give the proof of Theorem~\ref{theorem_2} as follows. 
\begin{proof}[Proof of Theorem~\ref{theorem_2}]
    From Lemma~\ref{lemma_2_2}, to determine whether a hypergraph $H$ can be immersed in $G$ is equivalent to testing if there is $\Tilde{H}\in\mathcal{D}(H)$ such that the densified factor graph $\Tilde{H}_{L}''$ has an embedding in the densified, $M$-generalised factor graph $G_{M,L}''$.  From Theorem~\ref{the:embedding-time}, there is an $O(|V(G_{M,L}'')|)$ time algorithm for the testing of one densified factor graph. 
    From Lemma~\ref{lem:hyperH-division-finite}, the set $\mathcal{D}(H)$ is finite and independent of $V(G)$.
    Therefore, the running time for testing immersion of a hypergraph $H$ is still $O(\left|V(G)\right|^6)$.
\end{proof}

\section{Dual hypergraph immersion}
\label{sec:dual}
Every hypergraph intrinsically supports a corresponding dual hypergraph known as the \emph{transpose}~\cite{robertson-seymour_rs10}.
Let $G$ be a loopless hypergraph. Its {transpose} is the hypergraph $G^\top$ where $V(G^\top) = E(G)$ and $E(G^\top) = V(G)$, with the same incidence relation as $G$. Following the definition of hypergraph immersion (Definition~\ref{definition_1}), we define the notion of \emph{dual immersion} for hypergraphs:
\begin{definition}
\label{definition_3}
Let $G,H$ be loopless hypergraphs. 
A \emph{dual immersion} of $G$ to $H$ is a function $\eta$ with domain $V(H)\cup E(H)$, such that
    \begin{enumerate}
    \item $\eta(e)\in E(G)$ for all $e\in E(H)$, and $\eta(e_1)\neq \eta(e_2)$ for all distinct $e_1,e_2\in E(H)$; 
    \item for each vertex $v\in V(H)$, if $v$ has distinct incident hyperedges $e_1,e_2,\cdots$, then $\eta(v)$ is a connected subgraph in $G$ that includes $\eta(e_1),\eta(e_2),\cdots$. 
    \item for all $v_1,v_2\in V(H)$ that do not share any hyperedge, $V(\eta(v_1)\cap \eta(v_2))=\emptyset$;
\end{enumerate} 
\end{definition}

In other words, there is an injective mapping from hyperedges in $H$ to hyperedges in $G$ and from ``vertices'', or more precisely sets of edges that incident on the same vertex, in $H$, to \textit{vertex-disjoint} connected subgraphs in $G$.  
As in the definition in ordinary graphs in section~\ref{sec:preliminary} and our extension to hypergraphs in section~\ref{sec:hypergraph}, immersion can be equivalently defined through operations in hypergraphs. 
Correspondingly, we now define the \textit{vertex coalescence} operation for the dual immersion problem, while maintaining the dewetting operation as in Definition~\ref{definition_2}. We then give an alternative definition for dual immersion as follows. 
\begin{definition}
\label{definition_4}
    Let $G,H$ be loopless hypergraphs. 
    A dual immersion of $H$ in $G$ exists if and only if $H$ can be obtained from a subgraph of $G$ by the following two operations (Fig.~\ref{fig_operation-dual}). 
    \begin{itemize}
        \item \emph{Vertex Coalesce}:
        Contract two vertices that share at least one hyperedge, resulting in a new vertex that is incident to all hyperedges originally incident to the two vertices. 
        \item \emph{Dewet} (same as in Definition~\ref{definition_2}): Detach a hyperedge from one vertex that the hyperedge is incident on.
\end{itemize}
\end{definition}
\begin{figure}[h!]
    \centering
    \begin{minipage}[b]{120pt}
		\centering
		{\includegraphics[width=120pt]{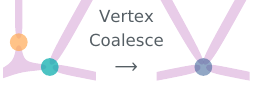}\subcaption{\label{fig_coalesce_vertex}}}
	\end{minipage}
    \hspace{30pt}
    \begin{minipage}[b]{120pt}
		\centering
		{\includegraphics[width=120pt]{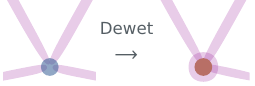}
        \subcaption{\label{fig_dewet_vertex}}}
	\end{minipage}
    \vspace{-3mm}
    \caption{\textbf{Dual immersion.}
    \subref{fig_coalesce_vertex}~The ``vertex coalescence'' operation involves merging two connected vertices, resulting in a new vertex that is incident to all hyperedges originally incident to the two vertices. \subref{fig_dewet_vertex}~The ``dewetting'' operation, same as in Fig.~\ref{fig_evaporate}, detaches a hyperedge from a vertex (some vertex in the lower left that is not shown here), thus reducing the hyperedge's size by one.
    \hfill\hfill}
    \label{fig_operation-dual}
\end{figure}
The vertex coalescence operation is related to the \textit{vertex identification}, or \textit{vertex contraction} operation, where two vertices $v_i, v_j$ in a graph are replaced by a single vertex $v$ such that $v$ is adjacent to the union of the vertices to which $v_i, v_j$ were originally adjacent.
However, they are intrinsically different, since vertex coalescence can only be applied to vertices that are directly connected by a hyperedge.
After vertex coalescence, all hyperedges remain intact and not removed, even if some have a size of one.

A potential application of dual immersion is in quantum memory networks~\cite{q-netw-memory_mlpnk24}. In such networks, distributed memories are represented by spatially distributed vertices. When two vertices establish a high-fidelity quantum communication channel, they are considered connected by an entanglement link, enabling the sharing of their total memories. This operation can be viewed as a form of vertex contraction between connected vertices, resulting in a new effective vertex that retains the incidence information of all its component vertices, thereby relating to the dual immersion problem.

It is worth noting that the duality between immersion and dual immersion guarantees that these two problems are fully equivalent, leading to the following theorem:
\begin{theorem}
A hypergraph $H$ can be immersed in $G$ if and only if $H^\top$ can be dual-immersed in $G^\top$.
\end{theorem}
\begin{proof}
    Proven by Definitions~\ref{definition_1} and~\ref{definition_3}.
\end{proof}
This provides an alternative view of the immersion problem. More importantly, the transpose of an ordinary graph is necessarily a hypergraph if the graph is not a regular graph of degree $2$, hence the results in this paper also provide theoretical guarantees for the immersion problem in more general settings. 

\bibliographystyle{amsplain}
\bibliography{refs}
\end{document}